\documentclass[10pt,journal,a4paper]{IEEEtran}
\usepackage[T1]{fontenc}
\usepackage[utf8]{inputenc}
\usepackage{amsmath,amssymb,amsthm,booktabs,graphicx,microtype}
\usepackage[hidelinks]{hyperref}
\usepackage{url}

\hypersetup{
  pdftitle={Bid Lattices and the Value of Flexibility: A Granularity Ratio
            for Capacity Markets},
  pdfauthor={Brieuc Le Roux-Tardif}
}

\newcommand{\SubEligStandalone}{61.0}
\newcommand{\SubEligCapacity}{69.9}
\newcommand{\SubEligDisplaced}{17.7}
\newcommand{\SubEligBeta}{85.5}
\newcommand{\ContinuousRevenue}{113.14}
\newcommand{\IntegerRevenue}{108.60}
\newcommand{\AnnualBias}{4.54}
\newcommand{\RelativeBias}{4.2}
\newcommand{\CapitalisedBias}{38.9}
\newcommand{\MaxEligiblePower}{1.5}
\newcommand{\MaxEligibleBias}{11.6}
\newcommand{\AnchorSpread}{21.2}
\newcommand{\HalfHourBias}{404}
\newcommand{\LatticeCases}{42}
\newcommand{\CapacityIntegral}{86}
\newcommand{\PeakBias}{21.5}
\newcommand{\PeakRho}{0.505}
\newcommand{\BoundTightness}{19}
\newcommand{\ContSpot}{43.28}
\newcommand{\IntSpot}{39.80}
\newcommand{\ContFcr}{69.86}
\newcommand{\IntFcr}{68.80}
\newcommand{\GapSpot}{3.48}
\newcommand{\GapFcr}{1.06}
\newcommand{\SpotShareOfGap}{77}
\newcommand{\FcrShareOfGap}{23}
\newcommand{\ContPledge}{0.64}
\newcommand{\IntPledge}{0.58}
\newcommand{\InvarianceMaxDev}{a relative 8e-15}

\newcommand{\BootBlock}{30}
\newcommand{\BootDraws}{2000}
\newcommand{\BootLow}{3.6}
\newcommand{\BootHigh}{4.6}
\newcommand{\BetaYearly}{3.0\% in 2024, 4.6\% in 2025, 4.5\% in 2026}
\newcommand{\SolverDevCentral}{a relative 4e-10}

\newcommand{\eur}{EUR}
\newtheorem{proposition}{Proposition}
\newtheorem{corollary}[proposition]{Corollary}

\begin{document}

\title{Bid Lattices and the Value of Flexibility:\\
A Granularity Ratio for Capacity Markets}

\author{Brieuc Le~Roux-Tardif%
\thanks{The author is with IMT Nord Europe, Institut Mines-T\'el\'ecom,
Univ.\ Lille, Centre for Energy and Environment, F-59000 Lille, France
(e-mail: brieuclerouxtardif@gmail.com).}%
\thanks{Manuscript prepared 8 August 2026. The optimisation model, the price and
capacity series and the raw case-by-case results that generate every number in
this paper are versioned in two companion repositories; see the reproducibility
section.}}

\markboth{Preprint, August 2026}{Le Roux-Tardif: Bid Lattices and the Value of
Flexibility}

\maketitle

\begin{abstract}
Capacity markets award reserve in discrete increments, yet valuation models
almost always treat the awarded quantity as a continuous variable. We show that
the size dependence of the resulting error is organised by a dimensionless number: the
granularity ratio $\rho=\delta/P$, where $\delta$ is the award increment and
$P$ the installed power of the asset. The co-optimisation of energy arbitrage
and symmetric reserve is positively homogeneous, so, with market data and all
other operating parameters fixed, asset size enters per-megawatt revenue only
through $\rho$ and storage duration; sweeping installed power at a fixed
increment and sweeping the increment at a fixed power are the same experiment.
Three structural results follow. Rounding a continuous award down to the
lattice is always feasible, which bounds the valuation error by $\rho$ times the
revenue of selling one megawatt in every product. That upper bound therefore
falls as the inverse of asset size. The largest award the lattice permits is
$P\rho\lfloor 1/\rho\rfloor$, so up to half of a connection can be
structurally unsellable, the worst case being an asset sized just below twice
the increment. Pledging the entire connection forces a zero net grid position
for the whole product, which converts a partial capacity commitment into
exclusion from the energy market. On 943 days of French day-ahead and
frequency containment reserve settlement prices, a 1~MW two-hour battery
priced with a continuous award is worth \RelativeBias\%
more than the same asset restricted to the 1~MW French lattice, and
\SpotShareOfGap\% of that difference is displaced arbitrage rather than
foregone capacity. The bias is not monotone in asset size: the largest value on
our finite grid is \PeakBias\% at $\rho=\PeakRho$. Bid granularity is therefore a
market-design variable with a size-dependent incidence, and the measured gap is
simultaneously a modelling error, a sizing rule, and an upper bound on what a
standalone asset should pay an aggregator.
\end{abstract}


\section{Introduction}

\IEEEPARstart{A}{lmost} every published valuation of a flexible asset in a
capacity market represents the awarded quantity as a continuous decision
variable. The choice is natural. It removes the only integer variable that a
day of dispatch would otherwise need, it is fast, and the awarded quantity is
physically continuous. Actual
markets, however, award capacity on a lattice. French frequency containment
reserve (FCR) is procured in 1~MW increments with a 1~MW minimum
\cite{rte2026,entsoe2026}; comparable minimum sizes and increments exist across
European reserve products and capacity mechanisms. A standalone 0.99~MW battery
cannot offer 0.99~MW, and a 1.99~MW battery cannot offer 1.99~MW.

This paper treats the award increment as a first-class parameter rather than an
implementation detail. Its central object is the dimensionless \emph{granularity
ratio}
\begin{equation}
  \rho = \frac{\delta}{P},
  \label{eq:rho}
\end{equation}
where $\delta$ is the award increment and $P$ the installed power at the grid
connection. We show that the joint energy-and-reserve optimisation is positively
homogeneous, so, conditional on market data and all other operating parameters,
asset size enters revenue per megawatt only through $\rho$ and storage duration.
That observation reorganises the problem. A continuous
model is the limit $\rho\to0$; when the minimum bid equals the increment, assets
with $\rho>1$ are excluded; and the entire family of results that a valuation study normally
reports as sensitivity to installed power is, exactly, a sweep in $\rho$.

Three structural consequences follow, each with a one-line proof and each
verified numerically. First, rounding a continuous award down onto the lattice
preserves feasibility of the whole schedule, which bounds the valuation error by
$\rho\Lambda$, where $\Lambda$ is the revenue of selling one megawatt of
capacity in every product of the horizon. The error is therefore $O(\rho)$ and
vanishes at least as fast as the inverse of asset size. Second, the largest award the
lattice permits is $\delta\lfloor P/\delta\rfloor$, so a fraction
$u(\rho)=1-\rho\lfloor1/\rho\rfloor$ of the connection is structurally
unsellable; $u$ is a sawtooth whose supremum over eligible assets is one half,
approached just below twice the increment. Third, an asset that pledges its entire
connection power is forced to a zero net grid position for the full duration of
the product, so at $\rho=1$ the bidder chooses between capacity and energy
rather than combining them.

The empirical section instantiates the theory on the French market with 943
days of day-ahead and FCR settlement prices from 2024 to 2026, a dense sweep of
$\rho$, a direct test of the homogeneity result, and a block bootstrap on the
central estimate. The measured gap at $\rho=1$ is \RelativeBias\%, of which
\SpotShareOfGap\% is displaced arbitrage rather than foregone capacity, which is
the opposite of the intuition that a coarse lattice mainly costs capacity
revenue.

The contribution is deliberately narrow in mechanism and broad in scope. We
propose no new dispatch algorithm and no new market. We identify the parameter
that controls a widespread modelling error, prove what can be proved about it,
measure what remains, and state the three uses of the resulting number: a
correction for modellers, a sizing rule for developers, and a distributional
diagnostic for market designers.

\section{Bid Lattices Among Non-Convexities}

Indivisibility is a classical source of difficulty in market design
\cite{scarf1994}. In electricity, the literature has concentrated on
indivisibilities of the \emph{technology}: start-up costs, minimum stable
generation and minimum up and down times make the unit commitment problem
non-convex, which is why linear prices do not support the efficient allocation
and why uplift, convex-hull pricing and their variants exist
\cite{oneill2005,liberopoulos2016,knueven2020}. European day-ahead coupling adds
indivisibility of the \emph{order}: block and fill-or-kill bids are integer
objects whose acceptance is decided by a mixed-integer program
\cite{madani2015}.

Bid granularity is a third kind, less studied and easier to overlook. Nothing
about a battery is indivisible: it can provide 0.37~MW of symmetric reserve as
easily as 1~MW. The indivisibility lives entirely in the product definition. It
is imposed by the market rather than by the technology, it applies uniformly to
all participants, and it therefore has a distributional incidence that depends
only on how large the participant is relative to the increment. That is the
object of this paper.

In the storage literature, the award increment appears occasionally as a
constraint but almost never as an estimand. Reviews of storage modelling
practice place bid granularity among the details routinely abstracted away
\cite{huneke2026}. Revenue-stacking studies establish that combined energy and
reserve operation raises value substantially
\cite{englberger2020,seifert2024,mirzaei2025}, and recent work adds degradation,
uncertainty and non-uniform bidding
\cite{brandle2026,hendrickx2026}. Dufo-L\'opez \emph{et al.} do impose an integer
FCR quantity when sizing a photovoltaic and battery plant \cite{dufolopez2024},
and market-compliant mixed-integer formulations encode technical requirements
directly \cite{mirzaei2025}. What is missing in all of them is the comparison
itself: none holds every other equation fixed and reports what the increment
costs. Physically, FCR provision is energy constrained even though capacity is
the remunerated product, so endurance and efficiency interact with any
quantity rule \cite{thien2017,hollinger2017}.

Our contribution relative to that body of work is a change of parametrisation.
Once the problem is written in $\rho$, the integrality gap of a
capacity-market bid becomes a one-dimensional object with provable structure,
and the classical rounding argument of integer programming
\cite{nemhauser1988} yields a bound that is informative rather than generic.

\section{Model}

We consider a single price-taking storage asset operating over one day, and we
repeat the day independently over the sample. Let $t\in\mathcal T$ index
energy-market intervals of length $\Delta t_t$ and let $b\in\mathcal B$ index
capacity products, each of which is a set of consecutive intervals. Write
$T_b=\sum_{t\in b}\Delta t_t$ for the elapsed duration of product $b$.

The asset is described by its grid-side power $P$, its usable energy $E$, and
its one-way efficiency $\eta$. Prices are the energy price $\pi_t$ and the
capacity price $\lambda_b$, both known ex post. The market requires that an
awarded megawatt be sustainable at full activation for an endurance $h$ in both
directions. Decision variables are charge $c_t$, discharge $d_t$, stored energy
$e_t$ at the end of interval $t$, and the symmetric capacity award $r_b$, held
constant within a product because that is what the market buys.

The program is
\begin{align}
  \max\ & \sum_{t}\pi_t(d_t-c_t)\Delta t_t+\sum_b \lambda_b r_b T_b
  \label{eq:obj}\\
  \text{s.t.}\ & e_t=e_{t-1}+(\eta c_t-d_t/\eta)\Delta t_t,
  \label{eq:soc}\\
  & 0\le c_t,d_t\le P,\qquad 0\le e_t\le E,
  \label{eq:box}\\
  & -P+r_b\le d_t-c_t\le P-r_b, & \forall t\in b,
  \label{eq:power}\\
  & r_b h/\eta\le e_{t-1},e_t\le E-r_b h\eta, & \forall t\in b.
  \label{eq:energy}
\end{align}
Constraint \eqref{eq:power} reserves headroom in power: capacity sold is power
that cannot be used for arbitrage in either direction. Constraint
\eqref{eq:energy} reserves headroom in energy, and is imposed on the stored
energy both before and after the scheduled action so that the model cannot award
reserve that only becomes deliverable after charging within the same interval.
Three further conditions close the day: charge and discharge are mutually
exclusive at every interval; discharged energy over the day is capped
at $\kappa$ equivalent full cycles; and stored energy returns to a fixed
fraction $\sigma$ of $E$ at the start and end of every day, so that no day
borrows inventory from its neighbour.

Everything above is standard. The object of this paper is the single remaining
constraint, the domain of the award:
\begin{equation}
  r_b\in\delta\mathbb Z_+\cap[0,P]
  \qquad\text{with}\qquad \delta\ge0,
  \label{eq:lattice}
\end{equation}
where $\delta=0$ is read as the continuous case $r_b\in[0,P]$. For
$\delta>0$, this formulation sets the minimum positive award equal to its
resolution, as in the FCR Cooperation. Markets with a distinct minimum require
a second dimensionless ratio. Let $R(\rho)$ denote the optimal value of
\eqref{eq:obj}--\eqref{eq:lattice}
per megawatt of installed power and per year, and define the relative
valuation bias of the continuous relaxation
\begin{equation}
  \beta(\rho)=\frac{R(0)-R(\rho)}{R(\rho)}.
  \label{eq:beta}
\end{equation}
Capacity prices satisfy $\lambda_b\ge0$. Finally let
\begin{equation}
  \Lambda=\sum_b\lambda_b T_b
  \label{eq:lambda}
\end{equation}
be the capacity revenue that one megawatt would earn by being sold in every
product of the sample, annualised like every other figure in this paper.
$\Lambda$ is a property of the price series alone and requires no optimisation,
which is what makes the bound of Section~IV-B usable before any model is
solved.

\section{Structural Results}

All results in this section hold for any energy-price path, nonnegative capacity
prices and any feasible parameter set. They are properties of the formulation,
not of the French data.

\subsection{The problem depends on size only through \texorpdfstring{$\rho$}{rho}}

\begin{proposition}[Homogeneity]\label{prop:scale}
Scaling $P$, $E$ and $\delta$ by the same factor $\alpha>0$ scales the optimal
value of \eqref{eq:obj}--\eqref{eq:lattice} by $\alpha$, and scales every
optimal decision variable by $\alpha$.
\end{proposition}

\begin{proof}
Map any feasible point $(c,d,e,r)$ to $(\alpha c,\alpha d,\alpha e,\alpha r)$.
Constraints \eqref{eq:soc}--\eqref{eq:energy}, the cycle cap $\kappa E$ and the
boundary condition $\sigma E$ are all positively homogeneous of degree one in
$(c,d,e,r,P,E)$, and $\alpha\delta\mathbb Z_+=\alpha\,(\delta\mathbb Z_+)$, so
the map is a bijection between the feasible sets. The objective
\eqref{eq:obj} is linear and therefore scales by $\alpha$. The operating-mode
exclusion is invariant because it constrains only which of $c_t,d_t$ is
nonzero.
\end{proof}

\begin{corollary}\label{cor:rho}
Conditional on the price path, product structure and other dimensionless
operating parameters, per-megawatt revenue $R$ depends on asset size only
through $E/P$ and $\rho=\delta/P$. Consequently a sweep of installed power at a
fixed increment is identical to a sweep of the increment at a fixed power under
those fixed conditions.
\end{corollary}

For the French 1~MW increment, a 1~MW asset sits at $\rho=1$ and a 5~MW asset
at $\rho=0.2$. The parametrisation makes asset sizes comparable within a fixed
market specification; it does not make markets with different prices, endurance
or product rules equivalent.

\subsection{Rounding bounds the error at first order}

\begin{proposition}[Rounding bound]\label{prop:round}
For every $\rho>0$,
\begin{equation}
  0\ \le\ R(0)-R(\rho)\ \le\ \rho\,\Lambda ,
  \label{eq:bound}
\end{equation}
where $\Lambda$ is defined in \eqref{eq:lambda} and both sides are expressed
per megawatt of installed power.
\end{proposition}

\begin{proof}
The left inequality holds because the lattice feasible set is contained in the
continuous one. For the right inequality, let $(c^\star,d^\star,e^\star,
r^\star)$ be optimal for $\delta=0$ and set
$\bar r_b=\delta\lfloor r^\star_b/\delta\rfloor$. Each of $r_b$'s constraints,
\eqref{eq:power} and \eqref{eq:energy}, is relaxed when $r_b$ decreases, and no
other constraint involves $r_b$; hence $(c^\star,d^\star,e^\star,\bar r)$ is
feasible with the same energy schedule. Since $\bar r_b\le r^\star_b\le P$ and
$\bar r_b\in\delta\mathbb Z_+$, we have $\bar r_b\le\delta\lfloor
P/\delta\rfloor$, so $\bar r$ satisfies \eqref{eq:lattice}. The objective loses
only capacity revenue, and $r^\star_b-\bar r_b<\delta$ for every $b$, so
before normalisation, the objective loss is at most
$\sum_b\lambda_b T_b(r^\star_b-\bar r_b)<\delta\Lambda$. Dividing by $P$
gives $R(0)-R(\rho)<\rho\Lambda$ per megawatt.
\end{proof}

Three readings of \eqref{eq:bound} matter. As a modelling statement, the
valuation error of a continuous relaxation is $O(\rho)$: its upper envelope
vanishes as the inverse of asset size, and a modeller can bound it before
solving anything, since $\Lambda$
is a price aggregate. As a market-design statement, the worst-case incidence of
a uniform increment falls as $1/P$; the realised loss remains sawtoothed. As a
commercial statement, $R(0)-R(\rho)$ is a benchmark for the gross value of
pooling under this model, and $\rho\Lambda$ is its upper bound.

The bound is loose by construction, because it holds the energy schedule fixed
while the lattice-constrained optimum is free to re-optimise it. Its role is to
fix the order in $\rho$, not to predict the level; Section~\ref{sec:results}
reports how loose it is on real prices.

\subsection{Above \texorpdfstring{$\rho=1$}{rho = 1} the bias is an identity, not a measurement}

\begin{proposition}[Eligibility]\label{prop:elig}
If $\rho>1$ then $r_b=0$ for every $b$, so $R(\rho)=V(0)$, the pure-arbitrage
optimum, and
\begin{equation}
  \beta(\rho)=\frac{V(P)-V(0)}{V(0)}
  \qquad\text{for all }\rho>1 .
  \label{eq:elig}
\end{equation}
\end{proposition}

\begin{proof}
$\lfloor P/\delta\rfloor=0$ when $\delta>P$, so \eqref{eq:lattice} forces
$r_b=0$ and the program reduces to pure arbitrage. The right-hand side of
\eqref{eq:elig} then does not depend on $\rho$.
\end{proof}

Proposition~\ref{prop:elig} matters mainly as a warning about how to report. The
quantity in \eqref{eq:elig} is the value of access to a market divided by the
value of not having it, and it is constant across the whole ineligible region. It
is therefore not a measurement of granularity: it does not vary with the
increment, and it can be made arbitrarily large by choosing an asset whose
standalone arbitrage value is small. It also decomposes cleanly. Writing the
continuous optimum as arbitrage plus capacity, \eqref{eq:elig} becomes gross
capacity revenue minus the arbitrage the relaxation had to give up, all divided
by standalone arbitrage revenue. In the French sample, an ineligible two-hour
asset shows \SubEligCapacity{} k\eur{}/MW/year of capacity revenue against
\SubEligDisplaced{} k\eur{}/MW/year of displaced arbitrage, on a standalone base
of \SubEligStandalone{}, hence \SubEligBeta\%. We report the decomposition rather
than the headline ratio, because the ratio is a definition and the decomposition
is not.

\subsection{Part of the connection is structurally unsellable}

\begin{proposition}[Utilisation defect]\label{prop:defect}
The largest award the lattice permits is $\delta\lfloor P/\delta\rfloor
= P\rho\lfloor1/\rho\rfloor$. The unsellable fraction of the connection,
\begin{equation}
  u(\rho)=1-\rho\lfloor1/\rho\rfloor,
  \label{eq:defect}
\end{equation}
satisfies $u(\rho)=0$ if and only if $1/\rho$ is an integer, and
$\sup_{0<\rho\le1}u(\rho)=1/2$, approached as $\rho\downarrow1/2$.
\end{proposition}

\begin{proof}
Feasibility of \eqref{eq:lattice} requires $r_b/\delta$ to be a nonnegative
integer no greater than $P/\delta$, whose maximum is $\lfloor P/\delta\rfloor$.
On the branch $\rho\in(1/(n{+}1),1/n]$ we have $\lfloor1/\rho\rfloor=n$ and
$u(\rho)=1-n\rho$, which decreases from $1/(n{+}1)$ to $0$. The branch suprema
$1/(n{+}1)$ are maximal at $n=1$.
\end{proof}

Proposition~\ref{prop:defect} converts an empirical curiosity into a design
rule. The worst structural sizing on a lattice of step $\delta$ is a connection
power just below $2\delta$: nearly half of it cannot be offered as capacity.
An exact integer multiple of $\delta$ removes this utilisation defect, but not
necessarily the residual integrality gap. Whether changing connection size is
economical lies outside this model.

\subsection{At \texorpdfstring{$\rho=1$}{rho = 1} capacity and energy become exclusive}

\begin{proposition}[Eviction]\label{prop:evict}
If $r_b=P$, then $d_t-c_t=0$ for every $t\in b$, and product $b$ contributes
exactly zero arbitrage revenue whatever the prices.
\end{proposition}

\begin{proof}
Substituting $r_b=P$ into \eqref{eq:power} gives $0\le d_t-c_t\le0$. The
arbitrage term of \eqref{eq:obj} is a function of $d_t-c_t$ alone.
\end{proof}

The smallest eligible asset, $\rho=1$, has exactly two options in every product:
award nothing and trade energy freely, or award its full power and abstain from
the energy market for the entire product. The coarse lattice therefore does not
merely round the capacity quantity; it destroys the interior of the trade-off
between the two revenue streams. This is why, as Section~\ref{sec:results}
shows, most of the measured gap at $\rho=1$ is lost arbitrage rather than lost
capacity, and it is the reason the gap does not scale with the capacity price
alone.

\subsection{A lower bound from the shadow value of the award cap}

Proposition~\ref{prop:round} bounds the loss from above. The complementary
question, how much of the loss is unavoidable, has an equally short answer, and
it uses a quantity the continuous solve already produces.

Let $V(\kappa)$ be the optimal value of the continuous problem when the award is
capped at $r_b\le\kappa$ instead of $r_b\le P$, so that $V(P)$ is the continuous
optimum and $V(0)$ is the pure-arbitrage optimum. Let $\nu\ge0$ be a
supergradient of $V$ at $\kappa=P$, that is, the marginal annual revenue of one
additional megawatt of permitted award. $\nu$ is a dual variable of the same
linear program and costs nothing extra to obtain.

\begin{proposition}[Unavoidable loss]\label{prop:lower}
Assume the continuous program is a linear program, that is, the operating-mode
exclusion is slack at its optimum. Then for every $\rho>0$,
\begin{equation}
  R(0)-R(\rho)\ \ge\ u(\rho)\,\nu ,
  \label{eq:lower}
\end{equation}
per megawatt of installed power, with $u$ as in \eqref{eq:defect}.
\end{proposition}

\begin{proof}
Any solution feasible for the lattice problem satisfies
$r_b\le\delta\lfloor P/\delta\rfloor=(1-u)P$ and is feasible for the continuous
problem with that cap, so $R(\rho)\le V((1-u)P)$. The value of a linear program
is a concave function of its right-hand side, and $\kappa$ enters only as a
right-hand side, so $V((1-u)P)\le V(P)+\nu\bigl((1-u)P-P\bigr)=V(P)-\nu uP$.
Combining and dividing by $P$ gives \eqref{eq:lower}.
\end{proof}

Together, \eqref{eq:bound} and \eqref{eq:lower} sandwich the loss between a term
driven by the unsellable fraction and a term driven by the increment:
\begin{equation}
  u(\rho)\,\nu\ \le\ R(0)-R(\rho)\ \le\ \rho\,\Lambda .
  \label{eq:sandwich}
\end{equation}
The two ends have different economics. The upper bound is a price aggregate and
is known before optimising. The lower bound is a shadow price and is known after
one continuous solve, which is the solve a modeller was going to run anyway. The
practical reading of \eqref{eq:sandwich} is that a modeller who has solved the
easy problem already possesses a certificate on the error of having solved only
the easy problem. The certificate is void on any day where the operating-mode
exclusion binds, because the value function of a mixed-integer program need not
be concave in its right-hand side.

\subsection{Several products, several lattices}

Nothing above uses the fact that there is one capacity product family. If the
asset can sell $K$ products indexed by $k$, each on its own lattice $\delta_k$
and each with its own price integral $\Lambda_k$, the rounding argument applies
coordinate by coordinate, because rounding any single award down still relaxes
every constraint it appears in. The bound becomes
\begin{equation}
  R(0)-R(\rho_1,\dots,\rho_K)\ \le\ \sum_k \rho_k\Lambda_k ,
  \label{eq:multi}
\end{equation}
with $\rho_k=\delta_k/P$. The eviction mechanism of
Proposition~\ref{prop:evict} also survives, and it becomes sharper: with several
products competing for the same connection, a full pledge on the coarsest
lattice excludes the asset not only from the energy market but from the finer
capacity products as well. A market that refines the increment of one product
while leaving another coarse therefore does not necessarily help the small
participant, because the binding lattice is the coarsest one that the asset
wants to use. We do not test \eqref{eq:multi} here; the French FCR case is
$K=1$, and a multi-product test would require reserve price series whose
settlement conventions we have not audited to the same standard.

\subsection{The bias is not monotone in asset size}

\begin{proposition}[Divisibility monotonicity]\label{prop:mono}
If $\delta'$ divides $\delta$ then $R(\delta/P)\le R(\delta'/P)$. Without a
divisibility relation, no ordering holds in general.
\end{proposition}

\begin{proof}
If $\delta'\mid\delta$ then $\delta\mathbb Z_+\subseteq\delta'\mathbb Z_+$, so
the feasible set for $\delta$ is contained in that for $\delta'$. For the second
claim, take $P=1$, zero energy-market value and a positive capacity price. Then
revenue is proportional to $\delta\lfloor1/\delta\rfloor$. The coarser step
$\delta=0.5$ sells the full connection, whereas the finer but non-nested step
$\delta'=0.4$ sells only $0.8P$. Thus neither numerical step size nor $\rho$
orders non-nested feasible sets.
\end{proof}

Two practical consequences follow. Reporting $\beta$ on a grid of asset sizes
gives a lower bound on the worst case, never the worst case itself, unless the
grid contains the local maxima that Proposition~\ref{prop:defect} locates just
above the reciprocals of the integers. And the common intuition that a bigger
asset is always closer to the continuous ideal is false pointwise; it is true
only for the envelope, by Proposition~\ref{prop:round}.

\section{Data and Experimental Design}

\subsection{Market data}

The common sample runs from 1 January 2024 to 31 July 2026 and contains 943
complete local days with exactly six capacity products each. French day-ahead
prices come from the Fraunhofer ISE Energy-Charts platform
\cite{energycharts2026}, with one upstream missing day filled from the same
French day-ahead series on the ENTSO-E Transparency Platform. Prices are hourly
until the Single Day-Ahead Coupling transition and quarter-hourly thereafter.
French FCR settlement capacity prices are parsed from the monthly official
workbooks distributed by regelleistung.net \cite{regelleistung2026}. Each quote
is stated per megawatt and per product; we divide by the elapsed UTC duration
rather than a nominal four hours, because the midnight product lasts three
hours at the spring clock change and five at the autumn change. One local day at
the end of the sample is dropped because the capacity series stops before the
price series does.

The study is ex post by construction. Prices are known perfectly, and the
price-taking quantity chosen at the observed clearing price is assumed awarded.
This removes forecast risk and award risk from the comparison, which is what
makes the comparison identify granularity alone, and it makes every level
reported here a historical upper bound rather than an operational forecast.

\subsection{Experiments}

Four families of runs, \LatticeCases{} cases in all, are solved on the same 943
days, all sharing prices, physics, endurance and the daily boundary condition.

\emph{Granularity sweep.} A dense grid of $\rho$ at a fixed 1~MW two-hour asset,
obtained by varying $\delta$ directly. This is the experiment
Corollary~\ref{cor:rho} says is the right one, and it resolves the local maxima
that a grid over installed power cannot see.

\emph{Homogeneity test.} The same $\rho$ realised at different absolute sizes,
which turns Proposition~\ref{prop:scale} into a falsifiable numerical claim
rather than an algebraic remark.

\emph{Duration interaction.} The granularity sweep repeated at one-hour and
four-hour durations, to check that the $\rho$ parametrisation is not an artefact
of the reference duration.

\emph{Robustness.} Every case at the reference point solved with two independent
solvers, and a block bootstrap over the 943 daily paired differences to attach
dispersion to $\beta$.

Table~\ref{tab:assumptions} lists the parameters held fixed throughout. Revenue
is annualised per installed megawatt over the common sample. Where an annual
figure is translated to a capital scale, we report $A\,\Delta R$ with
$A=(1-1.08^{-15})/0.08$, a fifteen-year eight-percent annuity factor. That is a
unit conversion of a revenue difference and not a project net present value:
capital cost, operating cost, degradation, taxes and residual value are
deliberately absent.

\begin{table}[t]
\centering
\caption{Parameters held fixed across all formulations.}
\label{tab:assumptions}
\begin{tabular}{ll}
\toprule
Parameter & Value \\
\midrule
Reference asset & 1 MW / 2 MWh \\
Round-trip efficiency & 85\% \\
Cycle cap $\kappa$ & 1.5 equivalent cycles/day \\
Daily boundary $\sigma$ & 50\% of $E$, start and end \\
Capacity products & six local 4 h blocks/day \\
Endurance $h$ & 0.25 h \\
Granularity ratio $\rho$ & swept, $0$ to $1$ \\
Information & perfect foresight \\
Sample & 943 local days, 2024--2026 \\
\bottomrule
\end{tabular}
\end{table}

\section{Results}\label{sec:results}

\subsection{Homogeneity holds numerically}

Proposition~\ref{prop:scale} predicts that per-megawatt revenue is invariant
under joint rescaling of $P$, $E$ and $\delta$. Four independent solves at
$\rho=0.5$ and two-hour duration span $P=0.5$ to $4$~MW. Their objectives agree to
\InvarianceMaxDev, so the observed spread is floating-point roundoff and not a
modelling difference. The proposition is not merely true on paper; the
implementation respects it, which is the strongest single check available on
this pipeline, because a sign or scaling error in the capacity constraints would
almost certainly break the invariance while leaving each individual solve
plausible. The four rows remain in the generated result file.

\subsection{The bias as a function of \texorpdfstring{$\rho$}{rho}}

Fig.~\ref{fig:rho} is the central result. It reports $\beta(\rho)$ on the dense
grid, the rounding bound of Proposition~\ref{prop:round} above it, and the
utilisation defect $u(\rho)$ of Proposition~\ref{prop:defect} below.

Two structural features are visible and both are predicted. The curve is a sawtooth
whose local minima sit exactly at $\rho=1/n$, where the defect vanishes and the
asset can pledge its whole connection. Within each branch the bias grows as
$\rho$ falls towards the next reciprocal from above, because the largest
feasible award shrinks while the connection does not. The largest bias on the grid
is \PeakBias\% at $\rho=\PeakRho$, an asset sized just below twice the
increment, which is where Proposition~\ref{prop:defect} says the defect is
worst. This is a prediction and not a fit: the grid was refined at $\rho$ slightly
above $1/2$ precisely because Proposition~\ref{prop:defect} says the defect
approaches its maximum of one half there, and the measured bias duly reaches its
own maximum at the same place. It remains formally a lower bound on the
supremum, since the supremum is approached and not attained. A grid over
installed power at a fixed
1~MW increment, as used in the earlier version of this study, would have
reported \MaxEligibleBias\% at \MaxEligiblePower~MW, less than two thirds of the
value found here, which is the practical content of
Proposition~\ref{prop:mono}. The price integral is
$\Lambda=\CapacityIntegral$~k\eur{}/MW/year over this sample, so the bound of
Proposition~\ref{prop:round} at $\rho=1$ is that same figure.

\begin{figure}[t]
\centering
\includegraphics[width=\columnwidth]{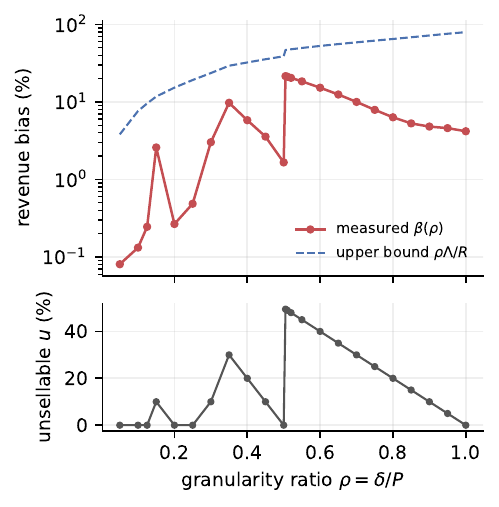}
\caption{Measured valuation bias $\beta(\rho)$ for a two-hour asset over 943
French days, on a logarithmic scale, with the rounding bound
$\rho\Lambda/R(\rho)$ of Proposition~\ref{prop:round}; below, the structurally
unsellable fraction $u(\rho)$ of Proposition~\ref{prop:defect}. Both curves fall
to a local minimum at every $\rho=1/n$; their proximity elsewhere is
descriptive, not an identity.}
\label{fig:rho}
\end{figure}

The full sweep is provided in \path{results/lattice_physical.csv}. Reading bias and
defect together separates two mechanisms. Where the defect is zero the
residual bias is small and comes only from the lumpiness of the award, never
exceeding what a single unclaimed increment can be worth. Where the defect is
large, it dominates: the asset owns power it is not allowed to offer, and no
re-optimisation recovers it.

The rounding bound is valid everywhere and loose by a factor of
\BoundTightness{} at the reference point $\rho=1$. That gap between bound and
realisation is itself informative: it measures how much the lattice-constrained
asset recovers by re-optimising its energy schedule, which the bound forbids by
construction.

\subsection{At \texorpdfstring{$\rho=1$}{rho = 1}, most of the loss is displaced arbitrage}

Table~\ref{tab:central} decomposes the reference case. The continuous relaxation
reports \ContinuousRevenue{} k\eur{}/MW/year against \IntegerRevenue{} on the
French 1~MW lattice, a difference of \AnnualBias{} k\eur{}/MW/year or
\RelativeBias\%. Capitalised at the stated annuity factor, that is
\CapitalisedBias{} k\eur{}/MW.

\begin{table}[t]
\centering
\caption{Reference case, $\rho=1$ (1 MW asset, 1 MW increment, 2 h duration).
The lattice costs more arbitrage revenue than capacity revenue.}
\label{tab:central}
\begin{tabular}{@{}lrrr@{}}
\toprule
& Continuous & Lattice $\rho=1$ & Difference \\
& (k\eur{}/MW/y) & (k\eur{}/MW/y) & (k\eur{}/MW/y) \\
\midrule
Energy arbitrage & \ContSpot & \IntSpot & \GapSpot \\
Capacity & \ContFcr & \IntFcr & \GapFcr \\
\midrule
Total & \ContinuousRevenue & \IntegerRevenue & \AnnualBias \\
\bottomrule
\end{tabular}
\end{table}

\begin{figure}[t]
\centering
\includegraphics[width=\columnwidth]{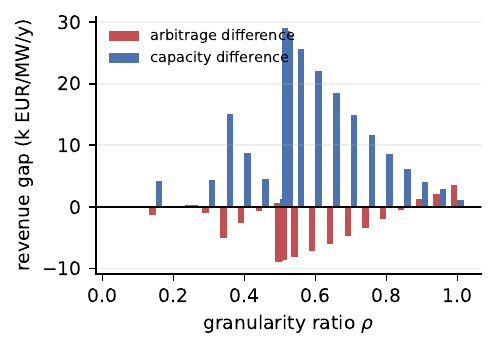}
\caption{Signed differences, continuous minus lattice, in the two accounting
lines. A negative bar means the lattice formulation earns more on that line;
the components are accounting differences, not separate causal effects.}
\label{fig:decomp}
\end{figure}

Fig.~\ref{fig:decomp} shows that the pattern is not specific to $\rho=1$.
Across the coarse part of the sweep the signed arbitrage difference is usually
larger. Some capacity differences are negative, so the two lines should not be
read as an additive causal attribution.

The decomposition is the empirical face of Proposition~\ref{prop:evict}.
Only \FcrShareOfGap\% of the gap is capacity revenue that the lattice made
unavailable. The remaining \SpotShareOfGap\% is energy arbitrage that the
lattice-constrained asset gave up in order to hold a full-power pledge, because
at $\rho=1$ a pledge is all or nothing and a full pledge freezes the net grid
position for four hours. The mean pledged fraction falls from
\ContPledge{} under a continuous award to \IntPledge{} on the lattice, and the
asset is worse off on both revenue lines at once. Any intuition that
sizes this effect by multiplying a capacity price by a rounded quantity is
therefore wrong by roughly a factor of four in this market.

\subsection{Dispersion}

A single price path yields a single realisation of $\beta$, so the point
estimate needs dispersion before it can be quoted. Because both formulations are
solved day by day on the same days, the annual gap is a sum of 943 paired daily
differences, and a moving-block bootstrap \cite{kunsch1989} over those
differences is available at no additional solver cost. With blocks of
\BootBlock{} days and \BootDraws{} resamples, the 95\% interval on $\beta$ at
$\rho=1$ is \BootLow--\BootHigh\%. Splitting the sample by calendar year gives
\BetaYearly, which shows that the estimate is not carried by a single period
even though capacity prices trend over the sample. The interval is conditional
on the observed path and the chosen 30-day block length; it is not a forecast
interval for another market regime.

\subsection{Duration and the endurance boundary}

At the five common grid points $\rho\in\{0,0.25,0.5,0.75,1\}$, one-, two- and
four-hour assets retain the same broad ordering while longer duration lowers the
measured level at the coarse points. This sparse sensitivity check does not
establish the full sawtooth away from the two-hour dense sweep.

One boundary case deserves separate treatment because it is where efficiency and
granularity interact. Combining the two sides of \eqref{eq:energy}, a nonzero
award of one increment requires
\begin{equation}
  E\ \ge\ \delta h(\eta+\eta^{-1}).
  \label{eq:minenergy}
\end{equation}
With $h=0.25$~h and $\eta=\sqrt{0.85}$, one megawatt of award requires
0.5017~MWh. A nominal 1~MW/0.5~MWh asset therefore misses eligibility by
1.7~kWh, one third of one percent of its nameplate energy, and its measured bias
is \HalfHourBias\%. We report this as a boundary artefact and not as a result:
it is a statement about a knife-edge, it would vanish for an asset of
0.505~MWh, and it is sensitive to the rounding convention used for endurance.
Section~\ref{sec:limits} returns to it.

\subsection{Solver and implementation checks}

Every case of the sweep is solved with HiGHS at a relative optimality gap of
$10^{-9}$, six orders of magnitude below the smallest effect reported here, so
the ordering $R(0)\ge R(\rho)$ cannot be a tolerance artefact. The reference
case is additionally solved with CBC, an independent code: the two objectives
agree to \SolverDevCentral. Both sides of every ratio use the same solver,
which the earlier version of this study did not guarantee at the 0.5~h point,
where both sides of \eqref{eq:energy} touch and CBC becomes numerically
degenerate on one negative-price day. Three assertions run at assembly time and
abort the pipeline if violated: every case uses exactly 943 days; the continuous
objective is never below its lattice counterpart; and no asset with $\rho>1$
receives a capacity award. No number in this manuscript is typed by hand: the
manuscript reads a single generated macro file.

\section{Implications}

\subsection{For modellers}

A continuous capacity award represents a \emph{pool}; for a standalone asset it
can be materially biased. Proposition~\ref{prop:round} gives an $O(\rho)$ upper
envelope, whose materiality still depends on $\Lambda$ and the revenue scale.
The cheap discipline is to report $\rho$ alongside every valuation. A
threshold on $\rho$ alone would be unsound, because Proposition~\ref{prop:mono}
forbids monotonicity: in our sweep the bias at $\rho=0.15$ exceeds the bias at
$\rho=0.5$. The defensible rule uses the two quantities the theory supplies.
Report both formulations whenever the connection is not an integer multiple of
the increment, that is whenever $u(\rho)>0$, and in any case whenever
$\rho\Lambda$ is material against the revenue at stake. Both tests are
arithmetic on published quantities, and their outcome is a transparent error bar
rather than a hidden assumption.

\subsection{For developers}

Proposition~\ref{prop:defect} supplies a design diagnostic: an integer multiple
of the award increment eliminates the structurally unsellable tail. It does not
eliminate every integrality loss, and changing connection size may itself cost
money. The theory brackets the remaining loss by $u(\rho)\nu$ and
$\rho\Lambda$; the French sweep reaches \PeakBias\% near the worst structural
sizing, just below twice the increment.

The measured gap has a second commercial reading. It is the gross value of
aggregation: the revenue a standalone asset forgoes purely because it cannot
reach the finer effective lattice a portfolio enjoys. At $\rho=1$ in this market
that is \AnnualBias{} k\eur{}/MW/year, and it is an upper bound on a rational
aggregation fee, before the aggregator's own forecast error, portfolio
constraints and margin.

\subsection{For market designers}

The upper envelope of the incidence is $O(\delta/P)$ per megawatt. A
single increment therefore weighs most heavily on small assets,
which makes bid granularity a distributional instrument whether or not it was
intended as one. Two design levers follow directly. Reducing $\delta$ shrinks
the envelope proportionally; the realised loss is guaranteed not to rise when
the new lattice nests the old one. Its effect can be bounded ex ante from the
price series through $\rho\Lambda$. Allowing
aggregation is the substitute lever: it moves small assets to a smaller
effective $\rho$ without changing the product. The choice between them is a
policy question, but the magnitude at stake is computable before the fact, which
is not usually the case for market-access rules.

\subsection{Reading another market through \texorpdfstring{$\rho$}{rho}}

The structural screening rules transfer; the empirical bias curve does not.
Applying the bounds to a different product requires three numbers,
two of which are published and one of which is a price aggregate. The increment
$\delta$ and the minimum bid size come from the product specification. The
installed power $P$ is a property of the asset. The price integral $\Lambda$ is a
sum over the settlement price series, computed without any optimisation. From
these, $\rho=\delta/P$ gives the scale ratio, $u(\rho)$ gives the structurally
unsellable fraction, and $\rho\Lambda$ bounds the valuation error. Predicting
the realised bias still requires the target market's prices, product structure
and operating parameters; the French curve in Fig.~\ref{fig:rho} is not portable.

Two qualitative predictions follow and are worth stating because they are
testable elsewhere. A market whose minimum bid size exceeds the typical asset in
a target segment does not merely reduce that segment's revenue; by
Proposition~\ref{prop:elig} it removes the segment, and no amount of price
increase compensates, because the constraint is on the quantity and not on the
price. And in any market where the increment equals the minimum bid size, the
smallest eligible participant sits exactly at $\rho=1$, where
Proposition~\ref{prop:evict} applies, so the smallest participant is also the
one for whom capacity and energy are most nearly exclusive. That coincidence is
a design choice rather than a necessity: nothing prevents a market from setting a
minimum bid size well above its award increment, which would preserve
eligibility screening while removing the exclusivity.

\section{Limitations}\label{sec:limits}

The estimate isolates one modelling choice and should not be read beyond it.
Activation energy is not simulated: \eqref{eq:power} and \eqref{eq:energy}
guarantee a symmetric fifteen-minute band, but measured frequency deviations,
corrective trading, imbalance settlement and induced ageing are absent, and all
of them affect compliance and lifetime in practice \cite{thien2017,mirzaei2025}.
Prices are known perfectly, so realisable value is lower than every level
reported here. Each day is self-contained, so inter-day carry is forgone and the
daily boundary condition replaces a continuation value; a rolling horizon would
change levels, and the boundary target is itself worth
\AnchorSpread{} k\eur{}/MW/year across the range we tested, which is more than
the reference granularity gap. That sensitivity is reported separately precisely
so that it is not silently attributed to the lattice.

The implementation first solves the faster relaxation and checks every interval.
If any simultaneous charge and discharge remains, that day is re-solved with
global operating-mode binaries. The returned optimum therefore enforces mutual
exclusion everywhere, including at positive prices, without burdening days on
which the relaxation is already physically feasible.

Three limitations are specific to the claims of this paper. First, the
homogeneity of Proposition~\ref{prop:scale} fails as soon as any cost or
constraint is not homogeneous of degree one in size, which includes fixed
connection charges, size-dependent tariffs and non-proportional degradation
models; the $\rho$ parametrisation is exact for the program studied here and
approximate for a richer one. Second, the sample is short and capacity prices
trend within it, so the levels are historical and the bootstrap interval is
conditional on the observed price path rather than on the distribution that
generated it. Third, the 0.5~h boundary case turns on 1.7~kWh and should be read
as an illustration of \eqref{eq:minenergy}, not as a magnitude.

Finally, the aggregation reading of the gap is an upper bound on a gross value.
For a real pool, continuous awards become attainable only at portfolio scale and
only after aggregator fees, forecast error and portfolio constraints, none of
which is modelled.

\section{Conclusion}

Conditional on the market and operating specification, the size effect of bid
granularity is organised by the ratio of the award increment to installed power. Under
that parametrisation the valuation error of the customary continuous relaxation
is bounded at first order by the ratio times the capacity revenue of a
permanently sold megawatt; a computable fraction of the connection, up to one
half, is structurally unsellable and vanishes only when the connection is an
integer multiple of the increment; and an asset at the minimum bid size is
forced to choose between capacity and energy rather than combining them. On 943
days of French data, the last effect dominates: at $\rho=1$ the continuous
relaxation overstates revenue by \RelativeBias\%, and \SpotShareOfGap\% of that
error is arbitrage the lattice-bound asset had to abandon, not capacity it
failed to sell.

The practical rules are short. Report $\rho$ and $\rho\Lambda$. Solve both
formulations when that bound is material, especially when the connection is not
an integer multiple of the increment. Evaluate integer-multiple connection
sizes rather than assuming they are costless. Treat the difference as an error
bar on valuation, an upper benchmark for aggregation value, and a
measure of how much a market's own quantity rule costs its smallest
participants.

\section*{Reproducibility}

Every number in this paper is generated by the provided workflow. The
optimisation model, the French price and capacity series, the workbook parser
and the unit tests are public at
\url{https://github.com/brieuclerouxtardif-blip/bess-arbitrage-fr}, at commit
\texttt{7742ad9}; the experiment definitions, raw case-by-case results,
bootstrap and figure scripts are available from the author. The manuscript
reads all numerical values from a generated macro file rather than from typed
text.

\bibliographystyle{IEEEtran}
\bibliography{refs}

\end{document}